\RequirePackage{fix-cm}
\documentclass[smallextended]{svjour3}       
\smartqed  
\usepackage[utf8]{inputenc}
\usepackage[OT4]{fontenc}
\usepackage{amsmath}
\usepackage{amsfonts}
\usepackage{amssymb}
\usepackage{subfigure}
\usepackage{bbold}
\usepackage[textwidth=1.2in]{todonotes}

\usepackage{tikz}
\usetikzlibrary{arrows,automata}
\usepackage{hyperref}
\usepackage{mathtools}
\usepackage{graphicx}
\usepackage{here}

\newcommand{\ie}{{\it ie}.}
\newcommand{\ket}[1]{\ensuremath{|#1\rangle}}
\newcommand{\bra}[1]{\ensuremath{\langle#1|}}

\newcommand{\proj}[1]{\ensuremath{\ket{#1}\bra{#1}}}
\newcommand{\braket}[2]{\ensuremath{\langle{#1}|{#2}\rangle}}

\newcommand{\Tr}{\mathrm{Tr}}

\newcommand{\tr}{\Tr}

\newcommand{\1}{{\rm 1\hspace{-0.9mm}l}}

\begin{document}

\title{Conditional entropic uncertainty relations for Tsallis entropies
}
\subtitle{}


\author{Dariusz Kurzyk         \and
      {\L}ukasz Pawela			\and
      Zbigniew Pucha{\l}a
}


\institute{D. Kurzyk \at
				Institute of Theoretical and Applied Informatics, Polish Academy of 
				Sciences\\
				Ba\l{}tycka 5, 44-100 Gliwice, Poland\\
              \email{dkurzyk@iitis.pl}           
           \and
           {\L}. Pawela \at
              Institute of Theoretical and Applied Informatics, Polish Academy of 
              Sciences\\
              Ba\l{}tycka 5, 44-100 Gliwice, Poland
           \and
           Z. Pucha{\l}a \at
           Institute of Theoretical and Applied Informatics, Polish Academy of 
           Sciences\\
           Ba\l{}tycka 5, 44-100 Gliwice, Poland\\
           Faculty of Physics, Astronomy and Applied Computer Science 
           Jagiellonian University,\\ {\L}ojasiewicza 11,  30-348 Krak{\'o}w, 
           Poland
}

\date{Received: date / Accepted: date}

\maketitle

\begin{abstract}
The entropic uncertainty relations are a very active field of scientific
inquiry. Their applications  include quantum cryptography and studies of quantum
phenomena such as correlations and non-locality. In this work we find
entanglement-dependent entropic uncertainty relations in terms of the Tsallis
entropies for states with a fixed amount of entanglement. Our main result is
stated as Theorem~\ref{th:bound}. Taking the special case of von Neumann entropy
and utilizing the concavity of conditional von Neumann entropies, we extend our
result to mixed states. Finally we provide a lower bound on the amount of
extractable key in a quantum cryptographic scenario.
\end{abstract}

\section{Introduction}
Formulated by Heisenberg \cite{heisenberg1927anschaulichen}, the uncertainty 
relation gives insight into differences between classical and quantum 
mechanics. According to the relation, simultaneous measurements of some 
non-commuting observables of a particle cannot be predicted with arbitrary 
precision.  

Numerous studies over the uncertainty relations led to entropic formulation by
Bia{\l}ynicki-Birula and Mycielski \cite{bialynicki1975uncertainty}, as a sum of
two continuous Shannon entropies, for probability distributions of position and
momentum. As our goal is to consider general observables, let us choose two
Hermitian non-commuting operators $X$ and $Y$. The first uncertainty relation
that holds for a pair of arbitrary observables was derived by Deutsch
\cite{deutsch1983uncertainty}
\begin{equation}
H(X)+H(Y)\geq -2 \log \frac{1+c}{2}=B_D,
\end{equation}
where $H(X)$ and $H(Y)$ denote the Shannon entropies of the probability 
distributions obtained during measurements of $X$ and $Y$ respectively. 
If $\ket{\phi_j}$, $\ket{\psi_k}$ are the eigenvectors of $X$ and $Y$, then 
$c=\max_{j,k}|\braket{\psi_j}{\phi_k}|$. 
Maassen and Uffink \cite{maassen1988generalized} obtained a stronger result 
\begin{equation}\label{massen}
H(X) + H(Y) \geq -2\log c=B_{MU},
\end{equation}
where $H(X)$, $H(Y)$ and $c$ are the same as in relation proposed by Deutsch. 

The entropic uncertainty relations are a very active field of scientific
inquiry \cite{wehner2010entropic,coles2017entropic}. One of the reasons are the
applications in quantum cryptography
\cite{koashi2005simple,divincenzo2004locking,damgaard2008cryptography}. Another
area where entropic uncertainty relations are widely used are studies of
quantum phenomena such as correlations and non-locality
\cite{guhne2004characterizing,oppenheim2010uncertainty,rastegin2016separability}.
Some results were generalized, hence entropic formulations of the uncertainty 
relation in terms of R{\'e}nyi entropies is included in 
\cite{zozor2013generalized}. Uncertainty relations for mutually unbiased bases 
and symmetric informationally complete measurements in terms of generalized 
entropies of R{\'e}nyi and Tsallis can be found \cite{rastegin2013uncertainty}.

In~\cite{kaniewski2014entropic} it was shown that entropic uncertainty
relations can be derived for binary observables from effective
anti-commutation, which can be important in device-independent cryptography.
This result was generalized in \cite{xiao2017uncertainty} for entropic
uncertainty relations in the presence of quantum memory.

The majorization-based bounds of uncertainty relation were first introduced by
Partovi in \cite{partovi2011majorization}, which was generalized in
\cite{puchala2013majorization,friedland2013universal}. In
\cite{puchala2013majorization}, majorization techniques were applied to obtain
lower bound of the uncertainty relation, which can give the bound stronger than
the well know result of Massen and Uffink. The formulation of strong
majorization uncertainty relation presented in \cite{rudnicki2014strong} is
involved, but in the case of qubits it can be expressed as
\begin{equation}\label{smeur}
H(X) + H(Y) \geq -c \log c - (1-c) \log (1-c) = B_{Maj2}.
\end{equation}
The asymptotic analysis of entropic uncertainty relations for random
measurements has been provided in~\cite{adamczak2016asymptotic} with the use
majorization bounds. Some interesting results along these lines are included in
\cite{rudnicki2015majorization,rastegin2016majorization,puchala2015certainty}.

In \cite{berta2010uncertainty}, Berta \emph{et al.} considered the uncertainty
relation for a system with the presence of a quantum memory. In this setup,
the system is described by a bipartite density matrix $\rho_{AB}$. Quantum conditional 
entropy can be defined as
\begin{equation}
S(A|B) = S(A, B) - S(B),\label{eq:chain}
\end{equation}
where $S(B)$ denotes the von Neumann entropy of the state $\rho_B=\tr_A 
\rho_{AB}$. Eq.~\eqref{eq:chain} is also known as the chain rule. We also 
introduce the states $\rho_{XB}$ and $\rho_{YB}$ as
\begin{equation}
\begin{split}
\rho_{XB} &= \sum_i \left(\proj{\psi_i} \otimes \1\right) \rho_{AB} 
\left(\proj{\psi_i} \otimes \1\right) \\
\rho_{YB} &= \sum_i \left(\proj{\phi_i} \otimes \1\right) \rho_{AB} 
\left(\proj{\phi_i} \otimes \1\right),
\end{split}
\end{equation}
which are post-measurement states, when the measurements were performed on the
part $A$. Berta \emph{et al.} \cite{berta2010uncertainty} showed that a bound
on the uncertainties of the measurement outcomes depends on the amount of
entanglement between measured particle and the quantum memory. As a
consequence, they formulated a conditional uncertainty relation given as
\begin{equation}\label{berta}
S(X|B) + S(Y|B) \geq B_{MU} + S(A|B)=B_{BCCRR}.
\end{equation}
Entropy $S(A|B)$ quantifies the amount of entanglement between the particle and
the memory. The bound of Berta \emph{et al.}~\cite{berta2010uncertainty} was
improved by Coles and Piani in \cite{coles2014improved} through replacing the
state-dependent value $B_{MU}$ with larger parameter. The result of Coles and
Piani was improved in \cite{xiao2016improved}. This relation was also generalized for R\'{e}nyi entropies and several important result can be found in \cite{coles2012uncertainty,muller2013quantum,tomamichel2014relating}. The uncertainty relation is also
considered in the context of \emph{quantum-to-classical randomness extractors}
(QC-extractors) \cite{berta2014quantum}. It is proved that QC-extractors gives
rise to uncertainty relation with the presence of a quantum memory.

In the absence of the quantum memory the bound (\ref{berta}) reduces to
(\ref{massen}) for pure $\rho_{AB}$. The results by Berta \emph{et al.}
\cite{berta2010uncertainty} can be applied to the problem of entanglement
detection \cite{li2011experimental} or quantum cryptography
\cite{coles2017entropic}. The bound quantified by Berta \emph{et al.}
\cite{berta2010uncertainty} was experimentally validated
\cite{prevedel2011experimental}.

In this paper we aim at finding state-independent entropic uncertainty relations
in terms of von Neumann and Tsallis entropies. Our results apply to states with
a fixed amount of entanglement, described by parameter $\lambda$. This allows us
to find non-trivial bounds for the entropic uncertainty relation. Otherwise we
would obtain a lower bound equal to zero. This bound is achieved in the case of
the maximally entangled state. Notice, that Berta \emph{et al.} formulated the
bound in a similar way. In their approach the information about entanglement was
hidden in the term of $H(A|B)$. In this case the bound is also zero for the
maximally entangled state.

JLet us now recall the notion of Tsallis
entropy~\cite{tsallis1988possible} which is a non-additive generalization of von
Neumann entropy and for a state $\rho_X$ it is defined as
\begin{equation}
T_q(X) = \frac{1}{q-1}\left(1-\sum_i \nu_i^q \right),
\end{equation}
where $\nu_i$ are the eigenvalues of $\rho_X$ and $q\in[0, \infty)$. Tsallis
entropy is identical to the Havrda-Charv{\'a}t structural
$\alpha$-entropy~\cite{havrda1967quantification} in information theory.  Note
that when $q\to 1$ we have $T_q(X) \to S(X)$. The chain rule applies to Tsallis
entropies, hence
\begin{equation}
T_q(A|B) = T_q(A, B) - T_q(B).
\end{equation}
We will use the following notation for Tsallis point entropy
\begin{equation}
t_q(x) = \frac{1}{q-1}\left(1- x^q - (1-x)^q\right).
\end{equation}
In the limit $q \to 1$ we recover
\begin{equation}
h(x) = \eta(x) + \eta(1-x),
\end{equation}
where $\eta(x)=-x\log x$.

\section{Qubit conditional uncertainty relations}

Without a loss of generality let us assume that we start with an entangled
state $\rho_{AB}=\proj{\psi_{AB}}$, where
$\ket{\psi_{AB}}=\sqrt{\lambda}\ket{00}+\sqrt{1-\lambda}\ket{11}$. In this
case, the parameter $\lambda$ describes the entanglement between the parties
$A$ and $B$. We chose the eigenvectors of $X$ and $Y$ as $\ket{\phi_i} =
O(\theta)\ket{i}$ and $\ket{\psi_i} = O(\theta+\epsilon)\ket{i}$, where
\begin{equation}
O(\theta) = 
\left(\begin{matrix}
\cos\theta & -\sin\theta  \\
\sin\theta & \cos\theta  
\end{matrix}\right)\in SO(2)\label{eq:rotation}
\end{equation}
is a real rotation matrix. Hence, instead of optimizing the uncertainty
relation over all possible states $\rho_{AB}$, we will instead optimize over
$\theta$. Hereafter we assume $\theta,\varepsilon \in [0, \pi/2]$. In this case
we have
\begin{equation}
c = \begin{cases}
|\cos \varepsilon|,& \varepsilon \leq \pi/4 \\
|\sin \varepsilon|,& \varepsilon > \pi/4.
\end{cases}
\end{equation}

It is important to notice, that we can restrict our attention to real rotation
matrices. This follows from the fact, that any unitary matrix is similar to real
rotation matrix. Matrices are similar, $U \sim V$, if for some permutation
matrices $P_1,P_2$ and diagonal unitary matrices $D_1,D_2$, we have $V=P_1D_1 U
D_2 P_2$~\cite{puchala2013majorization}. Next we note that the eigenvalues of
states $\rho_{XB}$ are invariant with respect to the equivalence relation.

We should also note here, that the two qubit scenario, simple as it is, may be
easily generalized to an arbitrary dimension of system $B$.

As we are interested in binary measurements, the states $\rho_{XB}$ and 
$\rho_{YB}$ are rank-2 operators. The non-zero eigenvalues of $\rho_{XB}$ can 
be easily obtained as
\begin{equation}
\begin{split}
\mu^{XB}_1 = & \lambda\sin^2(\theta) + 
(1-\lambda)\cos^2(\theta),\\
\mu^{XB}_2 = & \lambda\cos^2(\theta) + 
(1-\lambda)\sin^2(\theta).
\end{split}\label{eq:eigs}
\end{equation}
To obtain the eigenvalues of $\rho_{YB}$ we need to replace $\theta$ with 
$\theta+\varepsilon$.

\subsection{Analytical minima}
Using eigenvalues of $\rho_{XB}$ and $\rho_{YB}$, we arrive at
\begin{equation}
T_q(X|B) + T_q(Y|B) = t_q(\mu_1^{XB}) + t_q(\mu_1^{YB}) -2 
t_q(\lambda).\label{eq:exact}
\end{equation}
Let us perform detailed analysis on the case when $q\to 1$, \ie\ the von
Neumann entropy case. We get
\begin{equation}
S(X|B)+S(Y|B) = h(\mu_1^{XB})+h(\mu_1^{YB}) - 2h(\lambda)\label{eq:01}
\end{equation}
In order to obtain an uncertainty relation, we need to minimize this quantity
over the parameter $\theta$. This is a complicated task even in the case
$\lambda=0$ and has been studied earlier~\cite{bosyk2011comment}. We guess that
$\theta = \pi/2-\varepsilon/2$ is an extremal point of~\eqref{eq:exact}.
Unfortunately, this point is the global minimum only when
\begin{equation}
-c \tanh^{-1} ((1-2\lambda)c)+\frac{(2\lambda-1)(1-c^2)}{c^2(1-2\lambda)^2 -1} 
< 0.\label{eq:boundary}
\end{equation}
A numerical solution of this inequality is shown in Fig.~\ref{fig:ineq}. 
When this condition is satisfied, the uncertainty relation is 
\begin{equation}
\begin{split}
&S(X|B)+S(Y|B) \geq\\ 
\log 4+&\eta(1+c-2\lambda c)+\eta(1-c+2\lambda c) -2h(\lambda).
\end{split}\label{eq:eta}
\end{equation}
When the condition in Eq.~\eqref{eq:boundary} is not satisfied, our guessed
extreme point becomes a maximum and two minima emerge, symmetrically to
$\theta=\pi/2-\varepsilon/2$. The reasoning can be generalized to $T_q$ in a
straightforward, yet cumbersome way. The details are presented in
Appendix~\ref{sec:app}. The solutions of inequality~\eqref{eq:boundary} along
with inequality~\eqref{eq:boundary-tsallis} for various values of $q$ are shown
in~Fig.\ref{fig:ineq}

\begin{figure}
	\centering{\includegraphics{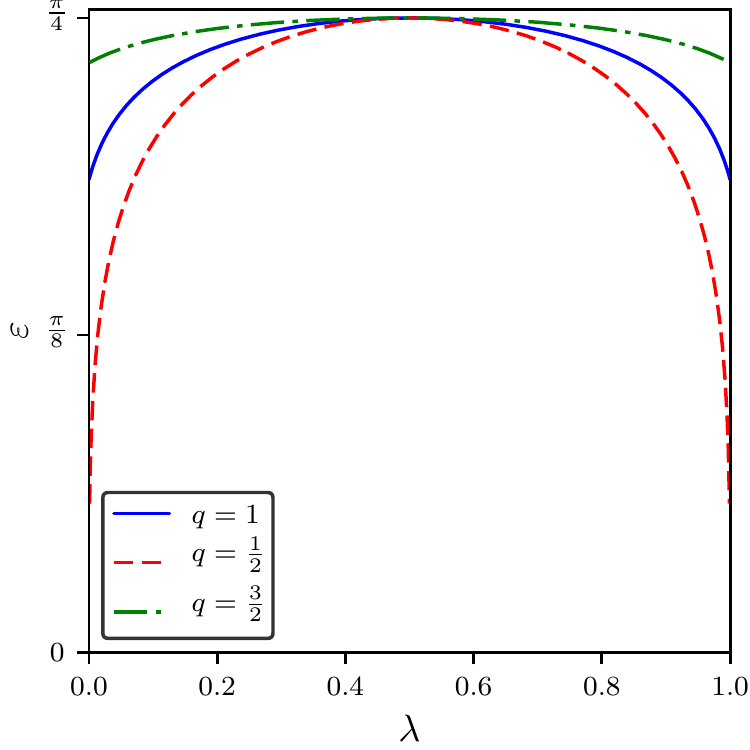}}
	\caption{Numerical solution of the inequality~\eqref{eq:boundary} ($q=1$) as a 
		function of $\lambda$ along with solution of a corresponding inequality for 
		chosen other values of $q$.}\label{fig:ineq}
\end{figure}

\subsection{Bounding the relative entropies}

In order to study the case of general Tsallis entropies $T_q$, we introduce the
following proposition

\begin{proposition}\label{th:bound}
	Let $\alpha \in [0,1]$ and $q \in [0, 2] \cup [3,\infty)$, then 
	\begin{equation}
	\begin{split}
	&t_q\big(\alpha p+(1-\alpha)(1-p)\big)\\ 
	&\geq 4 \frac{\left(\alpha ^q+(1-\alpha 
		)^q-2^{1-q}\right)}{\alpha ^q+(1-\alpha 
		)^q+q-2}  p(1-p) \big(1-t_q(\alpha)\big) + 
	t_q(\alpha).
	\end{split}
	\end{equation}
	In the cases $q=2$ and $q=3$ we have an equality.
\end{proposition}
\begin{proof}
	We define
	\begin{equation}
	\begin{split}
	&f(p) = t_q\big(\alpha p+(1-\alpha)(1-p)\big)\\ 
	&- 4 \frac{\left(\alpha
		^q+(1-\alpha)^q-2^{1-q}\right)}{\alpha ^q+(1-\alpha)^q+q-2} p(1-p)
	\big(1-t_q(\alpha)\big) - t_q(\alpha).
	\end{split}
	\end{equation}
	Next we note, that $f(0) = f(\frac12)=0$. We will show, that $f$ has no other
	zeros on interval $(0, \frac12)$. We calculate
	\begin{equation}
	\begin{split}
	&f{'''}(p)=
	(2 \alpha -1)^3 (q-2) q \big((\alpha -2 \alpha  p+p)^{q-3}\\
	&-(-\alpha +(2 
	\alpha -1) p+1)^{q-3}\big),
	\end{split}
	\end{equation}
	which is positive for $q\in [0,2)\cup(3,\infty)$ and $p \in[0, \frac12]$.
	Therefore we obtain, that $f'(x)$ is strictly convex on $(0, \frac12)$ and
	$f'(1/2)=0$.
	
	Now let us assume, that for $x_0 \in (0,1/2)$ we have $f(x_0) = 0$. Then by 
	Rolle's theorem, there exist points $0 < y_0 < x_0 < y_1 < \frac12$ such that 
	$f'(y_0) = f'(y_1) = 0$. Together with fact that $f'(1/2)=0$ we obtain  a 
	contradiction with the convexity of $f'$ on $(0,\frac12)$.
	
	Last thing to show is that  for some $\varepsilon \in (0,\frac12)$ we have 
	$f(\varepsilon) >0$.
	To show it we write 
	\begin{equation}
	\begin{split}
	f'(0) &= \frac{\alpha ^{q-1} (2 \alpha  (q-2)-q)}{q-1}\\
	&+\frac{(4 \alpha -2 \alpha  q+q-4) 
		(1-\alpha )^{q-1}+2^{3-q}}{q-1}
	=: g(\alpha).
	\end{split}
	\end{equation}
	Now we note, that $g(\alpha)$ is positive for $\alpha \in (0,1) \setminus
	\left\{\frac12  \right\}$ and $q \in [0,2)\cup (3,\infty)$. This follows form
	convexity of $g$ on these sets and the fact,  that it has a minimum,
	$g\left(\frac12\right) = 0$. From this fact there exist $\varepsilon>0$ such
	that $f(\varepsilon)>0$.
	
	The equalities in the case $q=2,3$ follow from a direct inspection.
\end{proof}

Now we are ready to state and prove the main result of this work

\begin{theorem}
	Let $\rho_{AB}=\proj{\psi_{AB}}$, where 
	$\ket{\psi_{AB}}=\sqrt{\lambda}\ket{00}+\sqrt{1-\lambda}\ket{11}$. Let us 
	choose two observables $X$ and $Y$ with eigenvectors $\ket{\phi_i} =
	O(\theta)\ket{i}$ and $\ket{\psi_i} = O(\theta+\epsilon)\ket{i}$, where 
	$O(\theta)$ is as in Eq.~\eqref{eq:rotation}. Then, the Tsallis entropic 
	conditional uncertainty relation is
	\begin{equation}
	T_q(X|B)+T_q(Y|B) \geq 
	2\frac{\lambda^q+(1-\lambda)^q-2^{1-q}}{\lambda^q+(1-\lambda)^q+q-2}
	(1-t_q(\lambda)) (1-c^2).\label{eq:main-result}
	\end{equation}
\end{theorem}

\begin{proof}
	Applying Proposition~\ref{th:bound}, to Eq.~\eqref{eq:exact} we get
	\begin{equation}
	\begin{split}
	&T_q(X|B)+T_q(Y|B)\geq\\ 
	&\frac{\lambda^q+(1-\lambda)^q-2^{1-q}}{\lambda^q+(1-\lambda)^q+q-2} 
	(1-t_q(\lambda)) (\sin^2(2\theta+2\epsilon)+\sin^2 2\theta)
	\end{split}
	\end{equation}
	The right hand side achieves a unique minimum $\theta=\pi/2-\varepsilon/2$ for 
	$\varepsilon\leq \pi/4$ and $\theta=\pi/4-\varepsilon/2$ for $\varepsilon > 
	\pi/4$. 
	Inserting this value we recover Eq.~\eqref{eq:main-result}.
\end{proof}

\begin{remark}
	In the limit $q\to 1$ we get the following uncertainty relation for Shannon 
	entropies
	\begin{equation}
	S(X|B)+S(Y|B) \geq 2(\log 2 - h(\lambda)) (1-c^2) = B_{KPP}.\label{eq:Bkpp}
	\end{equation}
\end{remark}

\begin{remark}
	Using the concavity of the conditional von Neumann entropy, we may generalize 
	bound~\eqref{eq:Bkpp} to mixed states $\rho_{AB}$. We get
	\begin{equation}
	S(X|B) + S(Y|B) \geq 2(\log 2 - S(B))(1-c^2).
	\end{equation}
\end{remark}

\begin{remark}
	The state dependent entropic uncertainty relation for $q\to 1$ reads
	\begin{equation}
	\begin{split}
	S(X|B) + S(Y|B) & \geq 2(\log 2 - h(\lambda))(\sin^2(2\theta + 2\varepsilon)+ 
	\sin^2 2\theta) \\
	&= B(\theta).\label{eq:02}
	\end{split}
	\end{equation}
\end{remark}
A comparison with the known entropic uncertainty relations for $\lambda=0$ and
$q\to 1$ is shown in Fig.~\ref{fig:separable}. As can be seen, our result gives
a tighter bound than the one obtained by Massen and Uffink for all values of
$\varepsilon$. The bound is also tighter than $B_{Maj2}$ when $\varepsilon$ is
in the neighborhood of $\pi/4$.

A comparison of the exact value~\eqref{eq:01}, state dependent lower bound
~\eqref{eq:02} and $B_{BCCRR}$ for different parameters $\lambda$, $\theta$ and
$\epsilon$ is presented in Figs~\ref{fig:example1} and~\ref{fig:example2}.

\begin{figure}\label{fig:separable}
	\centering\includegraphics{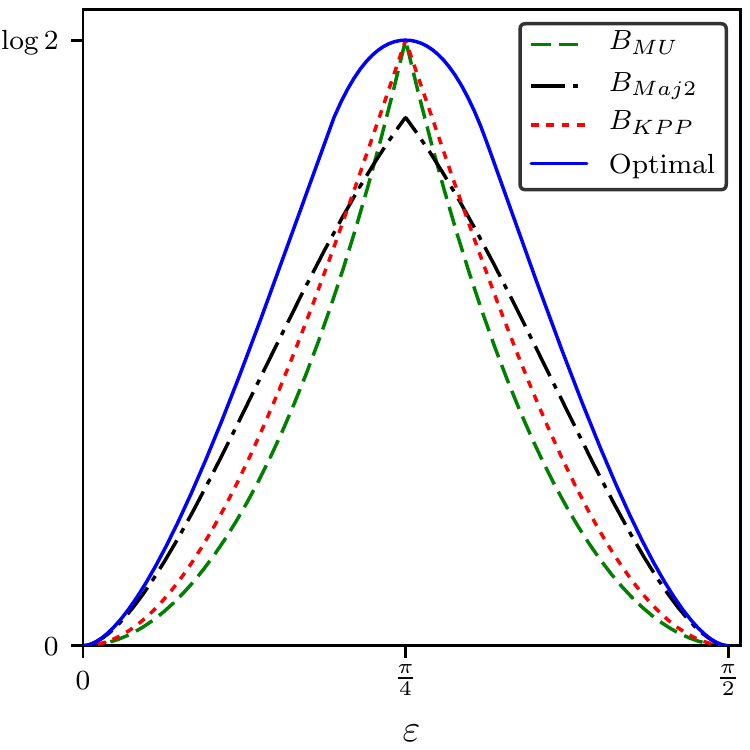}
	\caption{Comparison of our result with known bounds in the case $\lambda=0$. 
		Blue solid line is the (numerical) optimal solution, dashed green is the 
		$B_{MU}$ bound, black dashed-dotted is $B_{Maj2}$ and red dotted $B_{KPP}$.}
\end{figure}

\begin{center}
	\begin{figure}[H]\label{fig:example}
		\centering \subfigure[$\lambda=0.1$, $\epsilon=\pi/4.2$] {%
			\centering 	
			\includegraphics[scale=1]{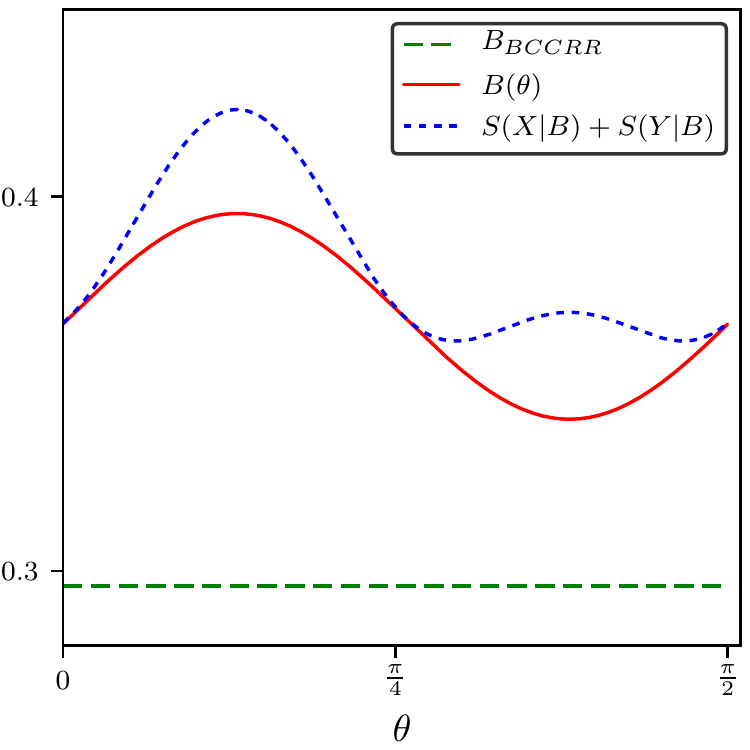} \label{fig:lam0.7} }%
		\subfigure[$\lambda=0.1$, $\epsilon=\pi/6$] {%
			\centering \includegraphics[scale=1]{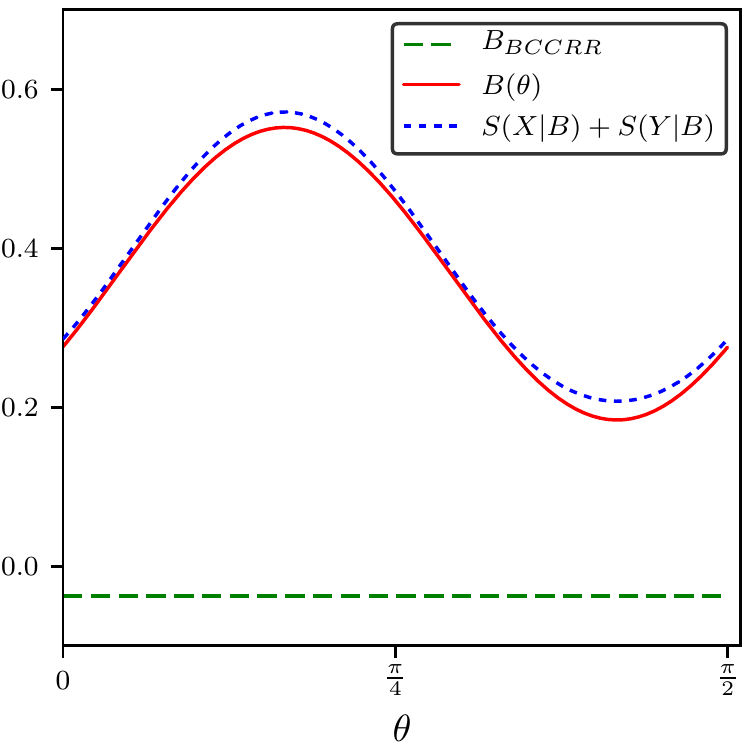} \label{fig:lam0.9} }%
		\caption{Comparison of our state-dependent result with $B_{BCCRR}$ and the
			exact value of conditional entropies for different parameters $\lambda$ and
			$\epsilon$ as a function of $\theta$.} %
		\label{fig:example1}
	\end{figure}
		
	\end{center}
\begin{figure}
	\centering\includegraphics{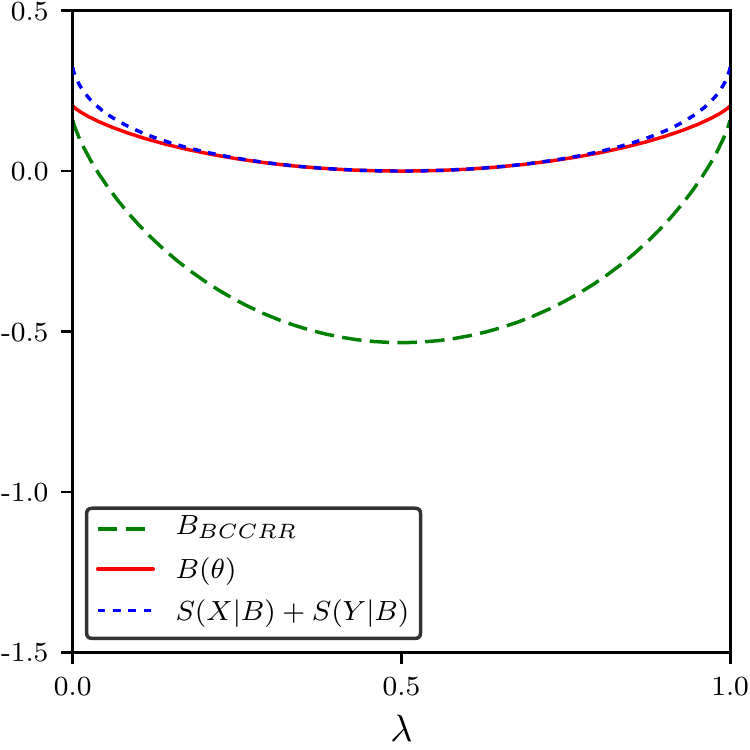}
	
	\caption{Comparison of $B(\theta)$ with $B_{BCCRR}$ and exact values of
		conditional entropies as a function of $\lambda$. Here $\varepsilon=\pi/8$,
		$\theta=\pi/2-\varepsilon/2$.}\label{fig:example2}
\end{figure}

\section{Security of quantum key distribution protocols}

One of the possible application of the uncertainty relation is quantum
cryptography, where the relation allows us to bound of the amount of key the
parties are able to extract per state. 

Assume that an eavesdropper creates a quantum system $\rho_{ABE}$. Next, parts
$A$ and $B$ are distributed to Alice and Bob. The generation of a secret key is
based on measurements $X,Y$ and $X',Y'$ performed by Alice and Bob,
respectively. Subsequently, Alice and Bob inform each other of their choices of
measurements. The security of the key depends on the correlation between the
measurement outcomes.

According to the investigations of Devetak and
Winter~\cite{devetak2005distillation}, the amount of extractable key is
quantified as $K\geq H(X|E)-H(X|B)$. Using our bound we are able to bound the
amount of extractable key in the terms of von Neumann entropies by
\begin{equation}\label{eqn:amount-key}
K \geq 2 (1-c^2)(\log 2 -S(B)) - S(A|B) - S(X|X') -S(Y|Y').
\end{equation}
In the above $S(X|X')$ is the conditional entropy of the state shared by Alice
and Bob, when both parties execute the measurement schemes $X,X'$ respectively.
This relates our result to~\cite{shor2000simple}. In our case Alice and Bob
need to upper bound entropies $S(A|B), S(X|X')$ and $S(Y|Y')$. The former
entropies can be bounded by quantities such as frequency of the agreement of
the outcomes.

\section{Conclusion} In this paper, we have derived new state-independent
uncertainty relations in terms of von Neumann and Tsallis entropies for qubits
and binary observables with respect to quantum side information. Our bounds
were compared with well-know bounds derived by Massen and Uffink
\cite{maassen1988generalized}, Rudnicki \emph{et al.} \cite{rudnicki2014strong}
and Berta \emph{et al.} \cite{berta2010uncertainty}. This paper can be also
treated as a generalization of results included in \cite{bosyk2011comment}.

Presented results are expected to have application to witnessing entanglement 
or in quantum cryptography as a measure of information in quantum key 
distribution protocols. Verification of our results in potential applications 
seems to be interesting task.

\begin{acknowledgements}
The authors acknowledge the support by the Polish
National Science Center under the Project Numbers 2013/11/N/ST6/03090 (D.~K.),
2015/17/B/ST6/01872 ({\L}.~P.) and 2016/22/E/ST6/00062 (Z.~P.).
\end{acknowledgements}


\section*{Appendix}
\appendix

\section{Generalization of Eq~\eqref{eq:exact} to Tsallis entropy 
case}\label{sec:app}
We start by introducing the following notation
\begin{equation}
\eta_q(x)=-\frac{x^q}{q-1}.
\end{equation}
Using this notation we note that
\begin{equation}
t_q(x)=\frac{1}{q-1}+\eta_q(x)+\eta_q(1-x)
\end{equation}
Now we recall Eq~\eqref{eq:exact}
\begin{equation}
T_q(X|B) + T_q(Y|B) = t_q(\mu_1^{XB}) + t_q(\mu_1^{YB}) -2 
t_q(\lambda).
\end{equation}
Again, we guess that the right hand side as a minimum at 
$\theta=\pi/2-\varepsilon/2$. Similar to the case $q\to 1$ we get that this is 
a minimum only when
\begin{equation}
(1+c-2\lambda c)^{q-2}\left( 2\lambda -1 +\frac{qc^2(1-2\lambda)+c}{q-1} 
\right)+
(1-c+2\lambda c)^{q-2}\left( 2\lambda -1 + \frac{qc^2(1-2\lambda)-c}{q-1} 
\right) > 0.
\label{eq:boundary-tsallis}
\end{equation}
This follows from the second derivative of Eq~\eqref{eq:exact} with respect to
$\theta$. The plots of the solutions to this inequality are shown in
Fig.~\ref{fig:ineq}. Note that when $q \to 1$ we recover the
bound~\eqref{eq:boundary}. When this is a minimum we obtain
\begin{equation}
T_q(X|B) + T_q(Y|B) \geq \frac{2}{q-1} +
2^{1-q}(\eta_q(1+c-2\lambda c) + \eta_q(1-c + 2\lambda c)) - 2 t_q(\lambda).
\end{equation}
In the case when $q\to 1$ we recover Eq~\eqref{eq:eta}.
\bibliography{cup}
\bibliographystyle{unsrt}

\end{document}